\documentclass[final]{siamltex}
\usepackage{amsmath,lmodern,ae,graphicx,tabularx,amsfonts,amssymb}
\font\dsrom=dsrom10 scaled 1200  
\newtheorem{remark}[theorem]{Remark}

\renewcommand{\P}{\mathbb{P}}
\newcommand{\indicatrice}{\textrm{\dsrom{1}}} 
\newcommand{\R}{{\mathbb R}}
\newcommand{\EE}{{\mathbb{E}}}
\newcommand{\NN}{{\mathbb N}}

\title{Error bounds for small jumps of L\'evy processes} 
\author{El Hadj Aly Dia\thanks{Universit\'e Paris-Est, Laboratoire d'Analyse et de Math\'ematiques Appliqu\'ees, UMR CNRS 8050, 5 bd. Descartes, Champs-sur-Marne, $77454$ Marne-la-Vall\'ee, France ({\tt dia.eha@gmail.com}).}}
\date{}
\begin{document}

\maketitle
  
\begin{abstract}
The pricing of options in exponential L\'evy models amounts to the computation
of expectations of functionals of L\'evy processes. In many situations, Monte-Carlo methods are used.
However, the simulation of a L\'evy process with infinite L\'evy measure generally requires either to truncate small jumps
or to replace them by a Brownian motion with the same variance. We will derive bounds for the errors generated by these two types of approximation. 
\end{abstract}
\begin{keywords} 
Approximation of small jumps, L\'evy processes, Skorokhod embedding, Spitzer identity
\end{keywords}
\begin{AMS}
60G51, 65N15
\end{AMS}
\begin{JEL}
C02, C15
\end{JEL}
\pagestyle{myheadings}
\thispagestyle{plain}
\markboth{E. H. A. DIA}{ERROR BOUNDS FOR SMALL JUMPS OF LEVY PROCESSES}
\section{Introduction}
\label{sec:intro}
In the recent years, the use of general L\'evy processes in financial models has grown extensively (see \cite{nielsen,cgmy,eberlein01}). A variety of numerical methods have been subsequently developed, in particular methods based on Fourier analysis (see \cite{broadie-yamamoto,feng-linetsky08,feng-linetsky,petrella-kou}). Nonetheless, in many situations, Monte-Carlo methods have to be used. The simulation of a L\'evy process with infinite L\'evy measure is not straightforward, except in some special cases like the Gamma or Inverse Gaussian models. In practice, the small jumps of the L\'evy process are either just truncated or replaced by a Brownian motion with the same variance (see \cite{asmussen-rosinski,cont-tankov,cont-voltchkova05,rydberg,signahl}). The latter approach was introduced by Asmussen and Rosinski \cite{asmussen-rosinski}, who showed that, under suitable conditions, the normalized cumulated small jumps asymptotically behave like a Brownian motion. 

The purpose of this article is to derive bounds for the errors generated by these two methods of approximation in the computation of functions of L\'evy processes at a fixed time or functionals of the whole path of L\'evy processes. We also derive bounds for the cumulative distribution functions. These bounds can be used to determine which type of approximations to use, since replacing small jumps by Brownian is more time-consuming (if we use Monte Carlo methods). Our bounds can be applied to derive approximation errors for lookback, barrier, American or Asian options. But this latter point will not be developed, and is left to another paper.

The characteristic function of a real L\'evy process $X$ with generating triplet $(\gamma,b^2,\nu)$ is given by
\begin{equation*}
  \EE e^{iuX_t}=\exp\left\{t\left(i\gamma u-\frac{b^2u^2}{2}+\int_{-\infty}^{+\infty}\left(e^{iux}-1-iux\indicatrice_{|x|\leq1}\right)\nu(dx)\right)\right\},
\end{equation*} 
where $\gamma\in\R$, $b\geq0$, and $\nu$ is a L\'evy measure. The process $X$ is the independent sum of a drift term $\gamma t$, a Brownian component $b B_t$, and a compensated jump part with L\'evy measure $\nu$. The process $X$ has finite (resp. infinite) activity if $\nu(\R)<\infty$ (resp. $\nu(\R)=+\infty$).

For $0<\epsilon\leq1$, the process $X^{\epsilon}$ is defined by
\begin{equation*}
X^{\epsilon}_t=\gamma t +b B_t+\sum_{0\leq s\leq t}\Delta X_s\indicatrice_{\left\{\left|\Delta X_s\right|>\epsilon\right\}}-t\int_{\epsilon<|x|\leq1}x\nu(dx).
\end{equation*}
The process $X^{\epsilon}$ is obtained (from $X$) by subtracting the compensated sum of jumps not exceeding $\epsilon$ in absolute value. Let
\begin{equation}\label{reps} 
R^{\epsilon}=X-X^{\epsilon}.
\end{equation}
The process $R^{\epsilon}$ is a L\'evy process with characteristic function
\begin{equation*}
  \EE e^{iuR^{\epsilon}_t}=\exp\left\{t\int_{|x|\leq\epsilon}\left(e^{iux}-1-iux\right)\nu(dx)\right\}.
\end{equation*} 
It holds $\EE\left(R^{\epsilon}_t\right)=0$ and $\mbox{Var}\left(R^{\epsilon}_t\right)=\sigma(\epsilon)^2t$, where
\begin{eqnarray*}
\sigma(\epsilon)=\sqrt{\int_{|x|\leq\epsilon}x^2\nu(dx)}.
\end{eqnarray*} 
Note that $\lim_{\epsilon\rightarrow0}\sigma(\epsilon)=0$. The behavior of $\sigma(\epsilon)$ when $\epsilon$ goes to $0$ is known for classical models (VG, NIG, CGMY...). As noted in Example 2.3 of \cite{asmussen-rosinski}, if $\nu(dx)=|x|^{-1-\alpha}L(x)dx$,
where $\alpha\in(0,2)$ and $L$ is slowly varying at $0$ , then it holds $\sigma(\epsilon)\sim \left(\left(L(-\epsilon)+L(\epsilon)\right)/(2-\alpha)\right)^{1/2}\epsilon^{1-\alpha/2}$; consequently, $\lim_{\epsilon \rightarrow 0}\sigma(\epsilon)/\epsilon=+\infty$.

We also define the process $\hat{X}^{\epsilon}$ by
\begin{equation*}
\hat{X}^{\epsilon}_t=X^{\epsilon}_t+\sigma(\epsilon)\hat{W}_t, \ t\geq0,
\end{equation*}
where $\hat{W}$ is a standard Brownian motion independent of $X$. We aim to study the behavior of the errors made by replacing $X$ by $X^{\epsilon}$ or $\hat{X}^{\epsilon}$, with respect to the level $\epsilon$. These errors are studied for the process $X$ at a fixed date and for its running supremum. Set, for any $t\geq0$,
\begin{equation*}
M_t=\sup_{0\leq s\leq t}X_s, \ M^{\epsilon}_t=\sup_{0\leq s\leq t}X^{\epsilon}_s, \ \hat{M}^{\epsilon}_t=\sup_{0\leq s\leq t}\hat{X}^{\epsilon}_s.
\end{equation*}
Unless stated otherwise, $X$ is a L\'evy process with generating triplet $(\gamma,b^2,\nu)$.

The paper is organized as follows. In the next section, we will study the errors resulting from the truncation of the compensated sum of small jumps. The results of that section are based on estimates for the moments of $R^\epsilon$. We also derive an estimate for
the expectation $\EE\left(M_t -M_t^{\epsilon}\right)$, by using Spitzer's identity. In Section~\ref{sec:brownianapproximation} we study the errors resulting from Brownian approximation. The process $X$ will be approximated by the process $\hat{X}^{\epsilon}$. A major result of Section~\ref{sec:brownianapproximation} is Theorem 2, which states an error bound for the expectation of a function of the supremum. This result is the consequence of Theorem~\ref{plongement}, which relies on the Skorohod embedding theorem.
\section{Truncation of the compensated sum of small jumps}
\label{sec:jumptruncation}
In this section, we will study the errors resulting from the approximation of $X$ by $X^{\epsilon}$. These errors are related to the moments of $R^{\epsilon}$. Define
\begin{equation}
\sigma_0(\epsilon)=\max\left(\sigma(\epsilon),\epsilon\right).
\end{equation}\label{sigma0}
The next result will be useful for many proofs in this paper.
\begin{proposition}\label{Rnprop}
Let $X$ be a L\'evy process and $R^{\epsilon}$ defined in \eqref{reps}. Then
\begin{eqnarray*}
\EE\left|R_{t}^{\epsilon}\right|^4&=&t\int_{|x|\leq\epsilon}x^4\nu(dx)+ 3\left(t\sigma(\epsilon)^2\right)^2,
\end{eqnarray*}
and for any real $q>0$
\begin{eqnarray*}
\EE\left|R_{t}^{\epsilon}\right|^{q}&\leq&K_{q,t}\sigma_0(\epsilon)^{q},
\end{eqnarray*}
where $K_{q,t}$ is a positive constant which depends only on $q$ and $t$.
\end{proposition}

\begin{proof}
Let $c_k\left(R_{t}^{\epsilon}\right)$ denote the k\emph{th} cumulant of $R_{t}^{\epsilon}$. Then $c_1\left(R_{t}^{\epsilon}\right)=\EE\left(R_{t}^{\epsilon}\right)=0$, and, for any $k\geq2$, $c_k\left(R_{t}^{\epsilon}\right)=t\int_{|x|\leq \epsilon}x^k\nu(dx)$ (note that $c_2\left(R_{t}^{\epsilon}\right)=\mbox{Var}\left(R_{t}^{\epsilon}\right)=\sigma^2(\epsilon)t$). See Proposition 1.2 of \cite{tankov}. Substituting into the general formula
\begin{eqnarray*}
\mu_4^{'}=c_4+4c_3c_1+3c_2^2+6c_2c_1^2+c_1^4
\end{eqnarray*}
(cf. \eqref{two} below), where, here and below, $\mu_k^{'}$ and $c_k$ denote the k\emph{th} moment and k\emph{th} cumulant of a distribution, respectively, gives the first part of the proposition. We now prove the second part. Let $n=\lceil q/2 \rceil$. Since $0<q/(2n)\leq1$,
\begin{eqnarray*}
\EE \left|R_t^{\epsilon}\right|^q&\leq&\left(\EE \left|R_t^{\epsilon}\right|^{2n}\right)^{\frac{q}{2n}}
\end{eqnarray*}
(by Jensen's inequality for concave functions). It thus suffices to prove the result for the case $q=2n$, $n\in\NN$; in fact, for any $n\in\NN$, it holds
\begin{equation}\label{eq_induction}
\left|\EE \left(R_t^{\epsilon}\right)^n\right|\leq K_{n,t}\sigma_0(\epsilon)^{n}. 
\end{equation}
The last inequality can be proved by induction as follows. It is trivial for $n = 0, 1, 2$. Suppose that \eqref{eq_induction} holds for all $n < m$. Then, by the well-known result (see e.g. Theorem 2 of \cite{morris})
\begin{equation}\label{two}
\mu_m^{'}=\sum_{n=0}^{m-1}\left(\begin{aligned}m-1\\ n\ \ \ \end{aligned}\right)\mu_n^{'}c_{m-n}, \ m\geq1,
\end{equation}
for all $m\geq2$ we have (recall that $c_1\left(R_t^{\epsilon}\right)=0$)
\begin{eqnarray*}
\left|\EE \left(R_t^{\epsilon}\right)^m\right|&\leq&\sum_{n=0}^{m-2}
\left(\begin{aligned}m-1\\ n\ \ \ \end{aligned}\right)\left|\EE \left(R_t^{\epsilon}\right)^n\right|\left|c_{m-n}\left(R_t^{\epsilon}\right)\right|.
\end{eqnarray*}
Hence, in view of the induction hypothesis, it suffices to show that $\left|c_{m-n}\left(R_t^{\epsilon}\right)\right|\leq t\sigma_0(\epsilon)^{m-n}$. Since $m-n\geq2$, we have $c_{m-n}\left(R_{t}^{\epsilon}\right)=t\int_{|x|\leq \epsilon}x^{m-n}\nu(dx)$, and hence
\begin{eqnarray*}
\left|c_{m-n}\left(R_t^{\epsilon}\right)\right|&\leq&t\int_{|x|\leq \epsilon}|x|^{m-n}\nu(dx)
\\&\leq&t\epsilon^{m-n-2}\int_{|x|\leq \epsilon}|x|^2\nu(dx)
\\&\leq&t\sigma_0(\epsilon)^{m-n}.
\end{eqnarray*}
The proposition is thus established.
\end{proof}

\subsection{Estimates for smooth functions}
Let $X$ be a L\'evy process and $f$ a $C$-Lipschitz function where $C>0$. Then,
\begin{eqnarray*}
 \EE \left|f\left(X_t\right)-f\left(X_t^{\epsilon}\right)\right|&\leq&C\EE\left|R_t^{\epsilon}\right|
\\&\leq&C\sqrt{\EE\left|R_t^{\epsilon}\right|^2}
\\&\leq& C\sqrt{t}\sigma(\epsilon).
\end{eqnarray*}
Note that we do not ask that $f\left(X_t\right)$ be integrable. If $f$ is more regular, sharper estimates can be derived, as shown in the following proposition.

\begin{proposition}\label{oX-}
Let $X$ be an infinite activity L\'evy process.
\begin{enumerate}
\item If $f\in C^{1}(\R)$ and satisfies $\EE\left|f^{'}\left(X_t^{\epsilon}\right)\right|<\infty$, and if there exists $\beta>1$ such that $\left(\sup_{\epsilon\in(0,1]}\EE\left|f^{'}\left(X_t^{\epsilon}+\theta R_t^{\epsilon}\right)-f^{'}\left(X_t^{\epsilon}\right)\right|^{\beta}\right)^{\frac{1}{\beta}}$ is finite and integrable with respect to $\theta$ on $[0,1]$, then
\begin{eqnarray*}
\EE\left(f\left(X_t\right) -f\left(X_t^{\epsilon}\right)\right)= o\left(\sigma_0(\epsilon)\right).
\end{eqnarray*}
\item If $f\in C^{2}(\R)$ and satisfies $\EE\left|f^{'}\left(X_t^{\epsilon}\right)\right|+\EE\left|f^{''}\left(X_t^{\epsilon}\right)\right|<\infty$, and if there exists $\beta>1$ such that $\left(\sup_{\epsilon\in(0,1]}\EE\left|f^{''}\left(X_t^{\epsilon}+\theta R_t^{\epsilon}\right)-f^{''}\left(X_t^{\epsilon}\right)\right|^{\beta}\right)^{\frac{1}{\beta}}$ is finite and integrable with respect to $\theta$ on $[0,1]$, then
\begin{eqnarray*}
\EE\left(f\left(X_t\right) -f\left(X_t^{\epsilon}\right)\right)= \frac{\sigma(\epsilon)^2t}{2}\EE f^{''}\left(X_t^{\epsilon}\right)+o\left(\sigma_0(\epsilon)^2\right).
\end{eqnarray*}
\end{enumerate}
\end{proposition}
Note that, if $f$ has bounded derivatives or $f$ is the exponential function and $e^{\beta X_t}$ is integrable, where $\beta>1$, the conditions in the above proposition are satisfied. Recall that the truncation of small jumps is used when $\nu(\R)=\infty$. In typical applications, we have $\liminf \sigma(\epsilon)/\epsilon>0$, so that $o\left(\sigma_0(\epsilon)^2\right)$ is in fact $o\left(\sigma(\epsilon)^2\right)$.

\begin{proof}
To prove part 1, we first write $f\left(X_t\right) -f\left(X_t^{\epsilon}\right)$ as
\begin{equation}\label{eqprop2}
 f\left(X_t\right) -f\left(X_t^{\epsilon}\right)=\int_{0}^{1}\left(f^{'}\left(X_t^{\epsilon}+\theta R_t^{\epsilon}\right)-f^{'}\left(X_t^{\epsilon}\right)\right)R_t^{\epsilon}d\theta
+f^{'}\left(X_t^{\epsilon}\right)R_t^{\epsilon}
\end{equation}
(by Theorem 27.4 of \cite{sato}, $R_t^{\epsilon}\neq0$ a.s.). Since $R_t^{\epsilon}$ and $X_t^{\epsilon}$ are independent, $\EE\left[f^{'}\left(X_t^{\epsilon}\right)R_t^{\epsilon}\right]=0$. For any $1<\alpha<\beta$, by H\"older's inequality,
\begin{equation*}
\EE\left|\left(f^{'}\left(X_t^{\epsilon}+\theta R_t^{\epsilon}\right)-f^{'}\left(X_t^{\epsilon}\right)\right)R_t^{\epsilon}\right|\leq
\left(\EE\left|f^{'}\left(X_t^{\epsilon}+\theta R_t^{\epsilon}\right)-f^{'}\left(X_t^{\epsilon}\right)\right|^{\alpha}\right)^{\frac{1}{\alpha}}\left(\EE\left| R_t^{\epsilon}\right|^{\frac{\alpha}{\alpha-1}}\right)^{\frac{\alpha-1}{\alpha}}.
\end{equation*}
By Lyapunov's inequality,
\begin{equation*}
\left(\EE\left|f^{'}\left(X_t^{\epsilon}+\theta R_t^{\epsilon}\right)-f^{'}\left(X_t^{\epsilon}\right)\right|^{\alpha}\right)^{\frac{1}{\alpha}}\leq\left(\EE\left|f^{'}\left(X_t^{\epsilon}+\theta R_t^{\epsilon}\right)-f^{'}\left(X_t^{\epsilon}\right)\right|^{\beta}\right)^{\frac{1}{\beta}}.
\end{equation*}
Further, the assumption $\sup_{\epsilon\in(0,1]}\EE\left|f^{'}\left(X_t^{\epsilon}+\theta R_t^{\epsilon}\right)-f^{'}\left(X_t^{\epsilon}\right)\right|^{\beta}<\infty$ implies that
the collection $\left\{\left|f^{'}\left(X_t^{\epsilon}+\theta R_t^{\epsilon}\right)-f^{'}\left(X_t^{\epsilon}\right)\right|^{\alpha}\right\}_{\epsilon\in(0,1]}$ is uniformly integrable; hence,
since \\$\left|f^{'}\left(X_t^{\epsilon}+\theta R_t^{\epsilon}\right)-f^{'}\left(X_t^{\epsilon}\right)\right|^{\alpha}\rightarrow 0$ a.s. as $\epsilon\rightarrow 0$, $\EE\left|f^{'}\left(X_t^{\epsilon}+\theta R_t^{\epsilon}\right)-f^{'}\left(X_t^{\epsilon}\right)\right|^{\alpha}\rightarrow 0$ (pointwise for $\theta\in[0,1]$). Therefore, by dominated convergence,
\begin{equation*}
\lim_{\epsilon\rightarrow 0}\int_0^1\left(\EE\left|f^{'}\left(X_t^{\epsilon}+\theta R_t^{\epsilon}\right)-f^{'}\left(X_t^{\epsilon}\right)\right|^{\alpha}\right)^{\frac{1}{\alpha}}d\theta=0.
\end{equation*}
Combined with Proposition~\ref{Rnprop}, it thus follows that
\begin{eqnarray*}
\int_{0}^{1}\EE\left[\left(f^{'}\left(X_t^{\epsilon}+\theta R_t^{\epsilon}\right)-f^{'}\left(X_t^{\epsilon}\right)\right)R_t^{\epsilon}\right]d\theta
=o\left(\sigma_0(\epsilon)\right).
\end{eqnarray*}
Part 1 of the proposition then follows from \eqref{eqprop2} (using Fubini's theorem).
We now prove the second part of the proposition. Using Taylor's formula we get
\begin{eqnarray*}
 \EE\left(f\left(X_t\right) -f\left(X_t^{\epsilon}\right)\right)&=& \EE\left[ f^{'}\left(X_t^{\epsilon}\right)\left(X_t-X_t^{\epsilon}\right)+\int_{X_t^{\epsilon}}^{X_t}f^{''}\left(x\right)(X_t-x)dx\right]
\\&=&\EE \left[f^{'}\left(X_t^{\epsilon}\right)R_t^{\epsilon}+\int_{0}^{1}f^{''}\left(X_t^{\epsilon}+\theta R_t^{\epsilon}\right)(1-\theta)\left(R_t^{\epsilon}\right)^2d\theta\right]
\\&=&\EE\left[\int_{0}^{1}f^{''}\left(X_t^{\epsilon}+\theta R_t^{\epsilon}\right)(1-\theta)\left(R_t^{\epsilon}\right)^2d\theta\right]
\\&=&\EE\left[\int_{0}^{1}f^{''}\left(X_t^{\epsilon}\right)(1-\theta)\left(R_t^{\epsilon}\right)^2d\theta\right]
\\&&+\EE\left[\int_{0}^{1}\left(f^{''}\left(X_t^{\epsilon}+\theta R_t^{\epsilon}\right)-f^{''}\left(X_t^{\epsilon}\right)\right)(1-\theta)\left(R_t^{\epsilon}\right)^2d\theta\right].
\end{eqnarray*}
The first expectation after the last equality sign is equal to $\frac{\sigma(\epsilon)^2t}{2}\EE f^{''}\left(X_t^{\epsilon}\right)$
while the second one can be shown to be $o\left(\sigma_0(\epsilon)^2\right)$ by following the proof of part 1. The proposition is proved.
\end{proof}

\begin{remark}\label{oX-lip}
Assume that $X$ is an integrable infinite activity L\'evy process and that $f\in C^{1}(\R)$ with $f^{'}$ being $C$-Lipschitz. Then
\begin{eqnarray*}
\left|\EE\left(f\left(X_t\right) -f\left(X_t^{\epsilon}\right)\right)\right|\leq\frac{C\sigma(\epsilon)^2t}{2}.
\end{eqnarray*}
\end{remark}
Indeed, $\EE\left[f^{'}\left(X_t^{\epsilon}\right)R_t^{\epsilon}\right]=0$ (by the assumptions on $X$ and $f$, $\EE\left|f^{'}\left(X_t^{\epsilon}\right)\right|<\infty$), and so the result follows directly from \eqref{eqprop2} using
\begin{eqnarray*}
\left|\EE\left(f\left(X_t\right) -f\left(X_t^{\epsilon}\right)\right)\right|\leq \EE\left[\int_0^1 \left|f^{'}\left(X_t^{\epsilon}+\theta R_t^{\epsilon}\right)-f^{'}\left(X_t^{\epsilon}\right)\right|\left|R_t^{\epsilon}\right|d\theta\right].
\end{eqnarray*}

We will consider now the case of the supremum process.

\begin{proposition}\label{OM}
Let $X$ be a L\'evy process and $f$ a $K$-Lipschitz function. Then
\begin{eqnarray*}
\EE\left|f\left(M_t\right) -f\left(M_t^{\epsilon}\right)\right|\leq 2K\sqrt{t}\sigma(\epsilon).
\end{eqnarray*}
\end{proposition}

\begin{proof}
We have
\begin{eqnarray*}
\EE\left|f\left(\sup_{0 \leq s \leq t} X_s\right) -f\left(\sup_{0 \leq s \leq t} X_s^{\epsilon}\right)\right|&\leq& K\EE\left|\sup_{0 \leq s \leq t} X_s -\sup_{0 \leq s \leq t} X_s^{\epsilon}\right|\\
&\leq& K\EE\sup_{0 \leq s \leq t} \left|R_s^{\epsilon}\right|\\
&\leq& K\sqrt{\EE\left(\sup_{0 \leq s \leq t} \left|R_s^{\epsilon}\right|\right)^2}.
\end{eqnarray*}
Note that $R^{\epsilon}$ is a c\`adl\`ag martingale. So, using Doob's inequality, we get
\begin{eqnarray*}
\EE\left|f\left(\sup_{0 \leq s \leq t} X_s\right) -f\left(\sup_{0 \leq s \leq t} X_s^{\epsilon}\right)\right|&\leq& 2K\sqrt{\EE\left|R_t^{\epsilon}\right|^2}\\
&=&2K\sqrt{t}\sigma(\epsilon).
\end{eqnarray*}
\end{proof}

\begin{remark}\label{OsupX}
Suppose that $X$ is an integrable L\'evy process and $f$ a function from $\R^{+}\times\R$ to $\R$, $K$-Lipschitz with respect to its second variable. Then
\begin{eqnarray*}
\left|\sup_{\tau\in\mathcal{T}_{[0,t]}}\EE f\left(\tau,X_{\tau}\right)-\sup_{\tau\in\mathcal{T}_{[0,t]}}\EE f\left(\tau,X_{\tau}^{\epsilon}\right)\right|\leq 2K\sqrt{t}\sigma(\epsilon),
\end{eqnarray*}
where $\mathcal{T}_{[0,t]}$ denotes the set of stopping times with values in $[0,t]$. For a proof, the reader is referred to \cite{dia}, pp. $67-68$.
\end{remark}

The bound in Proposition~\ref{OM} might not be optimal. This is what suggests the following result.
\begin{theorem}\label{petitssautstronc}
Let $X$ be an integrable infinite activity L\'evy process. Then
\begin{eqnarray*}
0\leq\EE\left(M_t-M_t^{\epsilon}\right)=o\left(\sigma(\epsilon)\right).
\end{eqnarray*}
\end{theorem}

\begin{proof}
Using Spitzer's identity (see Proposition $1$ in Section $3$ of \cite{dia-lamberton} for details), we have
\begin{eqnarray*}
\EE\left(M_t-M_t^{\epsilon}\right)&=&\int_0^t\frac{\EE X_s^{+}}{s}ds-\int_0^t\frac{\EE\left(X_s^{\epsilon}\right)^{+}}{s}ds
\\&=&\int_0^t\EE\left(X_s^{+}-\left(X_s^{\epsilon}\right)^{+}\right)\frac{ds}{s}.
\end{eqnarray*}
It holds
\begin{eqnarray*}
X_s^{+}-\left(X_s^{\epsilon}\right)^{+}&=&\left(X_s^{\epsilon}+R_s^{\epsilon}\right)^{+}-\left(X_s^{\epsilon}\right)^{+}
\\&=&\left(X_s^{\epsilon}+R_s^{\epsilon}\right)\indicatrice_{X_s^{\epsilon}+R_s^{\epsilon}>0}-X_s^{\epsilon}\indicatrice_{X_s^{\epsilon}>0}
\\&=&\left(X_s^{\epsilon}+R_s^{\epsilon}\right)\left(\indicatrice_{X_s^{\epsilon}>0}+\indicatrice_{X_s^{\epsilon}+R_s^{\epsilon}>0,X_s^{\epsilon}\leq0}-\indicatrice_{X_s^{\epsilon}+R_s^{\epsilon}\leq0,X_s^{\epsilon}>0}\right)-X_s^{\epsilon}\indicatrice_{X_s^{\epsilon}>0}
\\&=&\left(X_s^{\epsilon}+R_s^{\epsilon}\right)\left(\indicatrice_{X_s^{\epsilon}+R_s^{\epsilon}>0,X_s^{\epsilon}\leq0}-\indicatrice_{X_s^{\epsilon}+R_s^{\epsilon}\leq0,X_s^{\epsilon}>0}\right)+R_s^{\epsilon}\indicatrice_{X_s^{\epsilon}>0}
\\&=&\left(\left|R_s^{\epsilon}\right|-\left|X_s^{\epsilon}\right|\right)^{+}\left(\indicatrice_{X_s^{\epsilon}+R_s^{\epsilon}>0,X_s^{\epsilon}\leq0}+\indicatrice_{X_s^{\epsilon}+R_s^{\epsilon}\leq0,X_s^{\epsilon}>0}\right)+R_s^{\epsilon}\indicatrice_{X_s^{\epsilon}>0}.
\end{eqnarray*}
Set $I_s^{\epsilon}=\EE\left(X_s^{+}-\left(X_s^{\epsilon}\right)^{+}\right)$. Thus, since $\EE\left(R_s^{\epsilon}\indicatrice_{X_s^{\epsilon}>0}\right)=0$ (by independence),
\begin{eqnarray*}
0\leq I_s^{\epsilon}\leq \EE\left(\left|R_s^{\epsilon}\right|-\left|X_s^{\epsilon}\right|\right)^{+}.
\end{eqnarray*}
By the left inequality, $\EE\left(M_t-M_t^{\epsilon}\right)\geq0$. We now prove that $\EE\left(M_t-M_t^{\epsilon}\right)=o\left(\sigma(\epsilon)\right)$. Since $\left(\left|R_s^{\epsilon}\right|-\left|X_s^{\epsilon}\right|\right)^{+}\leq \left|R_s^{\epsilon}\right|\indicatrice_{\left|X_s^{\epsilon}\right|<\left|R_s^{\epsilon}\right|}$, we get $I_s^{\epsilon}\leq\EE\left( \left|R_s^{\epsilon}\right|\indicatrice_{\left|X_s^{\epsilon}\right|<\left|R_s^{\epsilon}\right|}\right)$. Hence, by Cauchy-Scwarz inequality, 
\begin{eqnarray*}
I_s^{\epsilon}&\leq&\left(\EE\left|R_s^{\epsilon}\right|^2\right)^{\frac{1}{2}}\left(\EE\left( \indicatrice_{\left|X_s^{\epsilon}\right|<\left|R_s^{\epsilon}\right|}\right)^2\right)^{\frac{1}{2}}
\\&=&\sigma(\epsilon)\sqrt{s}\P\left[\left|X_s^{\epsilon}\right|<\left|R_s^{\epsilon}\right|\right]^{\frac{1}{2}}.
\end{eqnarray*}
Thus,
\begin{eqnarray*}
0\leq\EE\left(M_t-M_t^{\epsilon}\right)&\leq&\sigma(\epsilon)\int_0^t\P\left[\left|X_s^{\epsilon}\right|<\left|R_s^{\epsilon}\right|\right]^{\frac{1}{2}}\frac{ds}{\sqrt{s}}.
\end{eqnarray*}
Since $\nu(\R)=\infty$, $R_s^{\epsilon}\rightarrow0$ a.s. and $X_s^{\epsilon}\rightarrow X_s$ a.s. with $X_s\neq0$. Hence $\P\left[\left|X_s^{\epsilon}\right|<\left|R_s^{\epsilon}\right|\right]^{\frac{1}{2}}\rightarrow0$ as $\epsilon\rightarrow0$.
Therefore, by dominated convergence,
\begin{eqnarray*}
\lim_{\epsilon\rightarrow0}\int_0^t\P\left[\left|X_s^{\epsilon}\right|<\left|R_s^{\epsilon}\right|\right]^{\frac{1}{2}}\frac{ds}{\sqrt{s}}=0,
\end{eqnarray*}
and so $\EE\left(M_t-M_t^{\epsilon}\right)=o\left(\sigma(\epsilon)\right)$.
\end{proof}

In financial applications, the function $f$ in Proposition~\ref{OM} is not always Lipschitz, as for call lookback option where the function is exponential. Hence the following proposition.
\begin{proposition}\label{OeM}
Let $X$ be a L\'evy process and $p>1$. If $\EE e^{pM_t}<\infty$, then
\begin{eqnarray*}
\EE \left|e^{M_t}-e^{M_t^{\epsilon}}\right|\leq C_{p,t}\sigma_0(\epsilon),
\end{eqnarray*}
where $C_{p,t}$ is a positive constant independent of $\epsilon$.
\end{proposition}

\begin{lemma}\label{OeMlemme}
Let $p>0$. If $\EE e^{pM_t}<\infty$ , then $\sup_{0<\delta\leq1}\EE e^{pM_t^{\delta}}<\infty$.
\end{lemma}

\begin{remark}\label{OeMlemmermq}
For any $p>0$, $\EE e^{pM_t}<\infty$ if and only if $\int_{x>1}e^{px}\nu(dx)<\infty$.
\end{remark}

The ``only if" part follows from Theorem 25.3 of \cite{sato}, noting that $e^{pX_t}\leq e^{pM_t}$. For the ``if" part, decompose $X$ as the independent sum $X = Y + Z + Z^{'}$ of L\'evy processes, where $Y$ has L\'evy measure $[\nu]_{\{|x|\leq1\}}$, and $Z$ and $Z^{'}$ are
pure jump with L\'evy measures $[\nu]_{\{x>1\}}$ and $[\nu]_{\{x<-1\}}$, respectively. Here $[\nu]_{E}$ denotes the restriction of $\nu$ to $E$. Note that $M_t\leq \sup_{0\leq s\leq t}Y_s+Z_t$; thus $\EE\left[ e^{pM_t}\right]\leq \EE\left[ e^{p\sup_{0\leq s\leq t}Y_s}\right]\EE\left[ e^{pZ_t}\right]$. It can be deduced from Theorems 25.3 and
25.18 of \cite{sato} that $\EE\left[ e^{p\sup_{0\leq s\leq t}Y_s}\right]$ is finite; so is $\EE\left[ e^{pZ_t}\right]$ by the former theorem,
under the assumption that $\int_{x>1}e^{px}\nu(dx)<\infty$. Hence $\EE\left[ e^{pM_t}\right]<\infty$.

\begin{proof}[Proof of Lemma~\ref{OeMlemme}] 
For $\delta\in(0,1]$, define $\bar{R}^{\delta}=X^{\delta}-X^1$. The process $\bar{R}^{\delta}$ is the compensated sum of jumps belonging to $(\delta,1]$ in absolute value. So
\begin{eqnarray*}
\EE e^{pM_t^{\delta}}&\leq&\EE e^{p\sup_{0\leq s\leq t}X_s^1+p\sup_{0\leq s\leq t}\bar{R}_s^{\delta}}
\\&\leq&\EE e^{p\sup_{0\leq s\leq t}X_s^1}\EE e^{p\sup_{0\leq s\leq t}\left|\bar{R}_s^{\delta}\right|}.
\end{eqnarray*}
By hypothesis and Remark~\ref{OeMlemmermq}, noting that Remark~\ref{OeMlemmermq} holds also for $M_t^1$, $\EE e^{p\sup_{0\leq s\leq t}X_s^1}<\infty$. We need to bound $\EE e^{p\sup_{0\leq s\leq t}\left|\bar{R}_s^{\delta}\right|}$ independently of $\delta$.
We have
\begin{eqnarray*}
\EE e^{p\sup_{0\leq s\leq t}\left|\bar{R}_s^{\delta}\right|}&=&\EE\sum_{n=0}^{+\infty}\frac{\left(p\sup_{0\leq s\leq t}\left|\bar{R}_s^{\delta}\right|\right)^n}{n!}
\\&=&1+p\EE\sup_{0\leq s\leq t}\left|\bar{R}_s^{\delta}\right|+\sum_{n=2}^{+\infty} \frac{p^n}{n!}\EE\left(\sup_{0\leq s\leq t}\left|\bar{R}_s^{\delta}\right|\right)^n.
\end{eqnarray*}
By Doob's inequality ($\bar{R}^{\delta}$ is a c\`adl\`ag martingale)
\begin{eqnarray*}
\EE e^{p\sup_{0\leq s\leq t}\left|\bar{R}_s^{\delta}\right|}&\leq&1+p\sqrt{\EE\left(\sup_{0\leq s\leq t}\left|\bar{R}_s^{\delta}\right|\right)^2}+\sum_{n=2}^{+\infty} \frac{p^n}{n!}\left(\frac{n}{n-1}\right)^n\EE\left|\bar{R}_t^{\delta}\right|^n
\\&\leq&1+2p\sqrt{\EE\left|\bar{R}_t^{\delta}\right|^2}+\sum_{n=2}^{+\infty} \frac{p^n}{n!}2^n\EE\left|\bar{R}_t^{\delta}\right|^n
\\&\leq&2p\sqrt{\mbox{Var}\left(\bar{R}_t^{\delta}\right)}+\EE\sum_{n=0}^{+\infty} \frac{p^n}{n!}2^n\left|\bar{R}_t^{\delta}\right|^n
\\&\leq&2p\sqrt{t\int_{\delta<|x|\leq1}x^2\nu(dx)}+\EE e^{2p\left|\bar{R}_t^{\delta}\right|}
\\&\leq&2p\sqrt{t\sigma(1)^2}+\EE e^{2p\bar{R}_t^{\delta}}+\EE e^{-2p\bar{R}_t^{\delta}}.
\end{eqnarray*}
It thus suffices to show that $\sup_{0< \delta\leq 1}\EE e^{\beta \bar{R}_t^{\delta}}<\infty$ for any $\beta\in\R$. Indeed, we have
\begin{equation*}
  \EE e^{\beta \bar{R}_t^{\delta}}=\exp\left\{t\int_{\delta<|x|\leq1}\left(e^{\beta x}-1-\beta x\right)\nu(dx)\right\}
\end{equation*} 
(a moment-generating function of a compensated compound Poisson process). By Taylor's theorem, $e^{\beta x}-1-\beta x=\beta^2x^2e^{\beta \xi}/2$ for any $|x|\leq1$, where $\xi$ is some number between 0 and $x$. This completes the proof, as it implies that
\begin{equation*}
  \EE e^{\beta \bar{R}_t^{\delta}}\leq\exp\left\{\frac{\beta^2t}{2}e^{|\beta|}\int_{|x|\leq1}x^2\nu(dx)\right\}.
\end{equation*} 
\end{proof}

\begin{proof}[Proof of Proposition~\ref{OeM}] 
By the mean value theorem, we have
\begin{eqnarray*}
e^{M_t}-e^{M_t^{\epsilon}}=\left(M_t-M_t^{\epsilon}\right)e^{\bar{M}_t^{\epsilon}},
\end{eqnarray*}
where $\bar{M}_t^{\epsilon}$ is between $M_t$ and $M_t^{\epsilon}$. Let $q$ be defined such that $\frac{1}{p}+\frac{1}{q}=1$.
\begin{eqnarray*}
\EE\left| e^{M_t}- e^{M_t^{\epsilon}}\right|&\leq&\EE\left|M_t-M_t^{\epsilon}\right|e^{\bar{M}_t^{\epsilon}}
\\&\leq&\EE\sup_{0\leq s\leq t}\left|R_s^{\epsilon}\right|e^{\bar{M}_t^{\epsilon}}
\\&\leq&\left(\EE\left(\sup_{0\leq s\leq t}\left|R_s^{\epsilon}\right|\right)^q\right)^{\frac{1}{q}}\left(\EE e^{p\bar{M}_t^{\epsilon}}\right)^{\frac{1}{p}}.
\end{eqnarray*}
Hence, using Doob's inequality and then Proposition~\ref{Rnprop}, we get
\begin{eqnarray*}
\EE\left| e^{M_t}- e^{M_t^{\epsilon}}\right|&\leq&\frac{q}{q-1}\left(\EE\left|R_t^{\epsilon}\right|^q\right)^{\frac{1}{q}}\left(\EE e^{p\bar{M}_t^{\epsilon}}\right)^{\frac{1}{p}}
\\&\leq&C_{p,t}\sigma_0(\epsilon)\left(\EE\left(e^{pM_t}+e^{pM_t^{\epsilon}}\right)\right)^{\frac{1}{p}},
\end{eqnarray*}
where $C_{p,t}$ denotes a constant depending on $p$ and $t$. We conclude the proof by Lemma~\ref{OeMlemme}.
\end{proof}

\subsection{Estimates for cumulative distribution functions}
For cumulative distribution functions, bounds are expected to be bigger. However, in some cases we can get similar results as in Lipschitz case. In the first result below, we assume local boundedness of the probability density function of the L\'evy process $X$ and its supremum process $M$ at a fixed time $t$. The regularity of the probability density function of a L\'evy process is studied in \cite{sato, bertoin}. For the supremum process see \cite{chaumont,dia}.
\begin{proposition}\label{OprobXM}
Let $X$ be a L\'evy process.
\begin{enumerate}
 \item If $b>0$, then
\begin{eqnarray*}
\sup_{x\in\R}\left|\P\left[X_t\geq x\right]-\P\left[X_t^{\epsilon}\geq x\right]\right|\leq \frac{1}{\sqrt{2\pi}b}\sigma(\epsilon).
\end{eqnarray*}
\item If $X_t$ has a locally bounded probability density function and $x\in\R$, then for any $q\in(0,1)$,
\begin{eqnarray*}
\left|\P\left[X_t\geq x\right]-\P\left[X_t^{\epsilon}\geq x\right]\right|\leq C_{x,t,q}\sigma_0(\epsilon)^{1-q},
\end{eqnarray*}
where, here and below, $C_{x,t,q}$ denotes a positive constant depending on $x$, $t$ and $q$.
\item If $M_t$  has a locally bounded probability density function on $(0,+\infty)$ and $x>0$, then for any $q\in(0,1/2)$,
\begin{eqnarray*}
\left|\P\left[M_t\geq x\right]-\P\left[M_t^{\epsilon}\geq x\right]\right|\leq C_{x,t,q}\sigma_0(\epsilon)^{1-q}.
\end{eqnarray*}
\end{enumerate}
\end{proposition}

\begin{lemma}\label{probX-Y}
Let $X$ and $Y$ be two r.v.'s. We assume that $X$ has a bounded density in a neighbourhood of $x\in\R$, and there exists $p>0$ such that $\EE\left|X-Y\right|^p$ is finite. Then there exists a constant $K_x>0$, such that for any $\delta>0$
\begin{eqnarray*}
\left|\P\left[X\geq x\right]-\P\left[Y\geq x\right]\right|\leq K_x\delta+\frac{\EE\left|X-Y\right|^p}{\delta^p}.
\end{eqnarray*}
\end{lemma}

\begin{proof}
We have
\begin{eqnarray*}
\left|\P\left[X\geq x\right]-\P\left[Y\geq x\right]\right|&=&\left|\P\left[X\geq x,Y<x\right]-\P\left[X<x,Y\geq x\right]\right|.
\end{eqnarray*}
We will study the above terms on the right of the equality.
\begin{eqnarray*}
\P\left[X\geq x,Y<x\right]&=&\P\left[x\leq X<x+\left(X-Y\right)\right]\\
&=&\P\left[x\leq X<x+\left(X-Y\right),\left|X-Y\right|\leq\delta\right] 
\\&&+\P\left[x\leq X<x+\left(X-Y\right),\left|X-Y\right|>\delta\right]\\
&\leq&\P\left[x\leq X<x+\delta\right] +\P\left[\left|X-Y\right|>\delta\right].
\end{eqnarray*}
Suppose that $X$ has a bounded density $f$ in the interval $[x-\delta_0,x+\delta_0]$, $\delta_0>0$ fixed, and let
\begin{eqnarray*}
K_x=\max\left\{\sup_{x-\delta_0\leq t\leq x+\delta_0}f(t),\frac{1}{\delta_0}\right\}.
\end{eqnarray*}
By considering the cases $\delta<\delta_0$ and $\delta\geq\delta_0$ separately, it is readily checked
that
\begin{eqnarray*}
\P\left[x\leq X<x+\delta\right]\leq K_x\delta,
\end{eqnarray*}
for any $\delta>0$. Thus, using Markov's inequality, we get
\begin{eqnarray*}
\P\left[X\geq x,Y<x\right]&\leq& K_x\delta+\frac{\EE\left|X-Y\right|^p}{\delta^p}.
\end{eqnarray*}
Similarly, using $\P\left[x-\delta\leq X<x\right]\leq K_x\delta$, it holds that 
\begin{eqnarray*}
\P\left[X<x,Y\geq x\right]&\leq& K_x\delta+\frac{\EE\left|X-Y\right|^p}{\delta^p}.
\end{eqnarray*}
Lemma~\ref{probX-Y} is thus established.
\end{proof}

\begin{proof}[Proof of Proposition~\ref{OprobXM}] 
We have
\begin{equation}\label{OprobXM_eq1}
\left|\P\left[X_t\geq x\right]-\P\left[X_t^{\epsilon}\geq x\right]\right|=\left|\P\left[X_t\geq x,X_t^{\epsilon}<x\right]-\P\left[X_t<x,X_t^{\epsilon}\geq x\right]\right|.
\end{equation}
It holds that
\begin{eqnarray*}
\P\left[X_t\geq x,X_t^{\epsilon}<x\right]&=&\P\left[x-\left(X_t-X_t^{\epsilon}\right)\leq X_t^{\epsilon}<x\right]\\
&=&\P\left[x-R_t^{\epsilon}\leq b B_t+\left(X_t^{\epsilon}-b B_t\right)<x\right].
\end{eqnarray*}
Note that $b B_t$ is independent of $X_t^{\epsilon}-b B_t$ and $R_t^{\epsilon}$, and $\frac{1}{\sqrt{2\pi t}b}$ is an upper bound of the probability density function of $b B_t$. Then, by conditioning on the pair $\left(R_t^{\epsilon},X_t^{\epsilon}-b B_t\right)$, it can be concluded that
\begin{eqnarray*}
\P\left[x-R_t^{\epsilon}\leq b B_t+\left(X_t^{\epsilon}-b B_t\right)< x\right]&\leq&\frac{1}{\sqrt{2\pi t}b}\EE\left|R_t^{\epsilon}\right|.
\end{eqnarray*}
Therefore, using that $\EE\left|R_t^{\epsilon}\right|\leq \sigma(\epsilon)\sqrt{t}$,
\begin{eqnarray*}
\P\left[X_t\geq x,X_t^{\epsilon}<x\right]&\leq&\frac{1}{\sqrt{2\pi}b}\sigma(\epsilon).
\end{eqnarray*}
Similarly
\begin{eqnarray*}
\P\left[X_t< x,X_t^{\epsilon}\geq x\right]&=&\P\left[x\leq X_t^{\epsilon}<x-\left(X_t-X_t^{\epsilon}\right)\right]\\
&=&\P\left[x\leq \sigma B_t+\left(X_t^{\epsilon}-\sigma B_t\right)<x-R_t^{\epsilon}\right]\\
&\leq&\frac{1}{\sqrt{2\pi}\sigma}\sigma(\epsilon).
\end{eqnarray*}
Hence part 1 of the proposition follows from \eqref{OprobXM_eq1}.

We now prove part 2 of the proposition. Let $p>0$. By Lemma~\ref{probX-Y} followed by Proposition~\ref{Rnprop}, there exist positive constants $K_{x,t}$ and $K_{p,t}$ such that
\begin{eqnarray*}
\left|\P\left[X_t\geq x\right]-\P\left[X_t^{\epsilon}\geq x\right]\right|&\leq& K_{x,t}\delta+\frac{\EE\left|X_t-X_t^{\epsilon}\right|^p}{\delta^p}
\\&=&K_{x,t}\delta+\frac{\EE\left|R_t^{\epsilon}\right|^p}{\delta^p}
\\&\leq& K_{x,t}\delta+K_{p,t}\frac{\sigma_0(\epsilon)^p}{\delta^p}
\end{eqnarray*}
for any $\delta>0$. Choosing $\delta=\sigma_0(\epsilon)^\frac{p}{p+1}$ yields
\begin{eqnarray*}
\left|\P\left[X_t\geq x\right]-\P\left[X_t^{\epsilon}\geq x\right]\right|&\leq&2\max\left(K_{x,t},K_{p,t}\right)\sigma_0(\epsilon)^\frac{p}{p+1},
\end{eqnarray*}
and so the result follows since $p/(p+1)$ can be chosen arbitrarily in $(0,1)$. 

We now prove part 3 of the proposition. Let $p>1$. By Lemma~\ref{probX-Y}, there exists a constant $K_{x,t}^{'}>0$ such that
\begin{eqnarray*}
\left|\P\left[M_t\geq x\right]-\P\left[M_t^{\epsilon}\geq x\right]\right|\leq K^{'}_{x,t}\delta+\frac{\EE\left|M_t-M_t^{\epsilon}\right|^p}{\delta^p}
\end{eqnarray*}
for any $\delta>0$. On the other hand
\begin{eqnarray*}
\EE\left|M_t-M_t^{\epsilon}\right|^p&\leq&\EE\left(\sup_{0\leq s\leq t}\left|X_s-X_s^{\epsilon}\right|\right)^p\\
&=&\EE\left(\sup_{0\leq s\leq t}\left|R_s^{\epsilon}\right|\right)^p.
\end{eqnarray*}
So by Doob's inequality, we have, using the constant $K_{p,t}$ from part 2,
\begin{eqnarray*}
\EE\left|M_t-M_t^{\epsilon}\right|^p&\leq&\left(\frac{p}{p-1}\right)^p\EE\left|R_t^{\epsilon}\right|^p
\\&\leq&K_{p,t}\left(\frac{p}{p-1}\right)^p\sigma_0(\epsilon)^p.
\end{eqnarray*}
Part 3 of the proposition then follows by choosing $\delta=\sigma_0(\epsilon)^{\frac{p}{p+1}}$.
\end{proof}
\section{Approximation of the compensated sum of small jumps by a Brownian motion}
\label{sec:brownianapproximation}
In this section we will replace $R^{\epsilon}$ by a Brownian motion. This method gives better results, subject to a convergence assumption. In fact, Asmussen and Rosinski proved (\cite{asmussen-rosinski}, Theorem 2.1) that, if $X$ is a L\'evy process, then the process $\sigma(\epsilon)^{-1}R^{\epsilon}$ converges in distribution to a standard Brownian motion, when $\epsilon\rightarrow0$, if and only if for any $k>0$
\begin{equation}\label{asmussen_rosinski_eq}
\lim_{\epsilon\rightarrow 0}\frac{\sigma\left(k\sigma(\epsilon)\wedge\epsilon\right)}{\sigma\left(\epsilon\right)}=1.
\end{equation}
Condition \eqref{asmussen_rosinski_eq} is implied by the condition 
\begin{equation}\label{asmussen_rosinski_rmq_eq}
\lim_{\epsilon\rightarrow 0}\frac{\sigma(\epsilon)}{\epsilon}=+\infty.
\end{equation}
The conditions \eqref{asmussen_rosinski_eq} and \eqref{asmussen_rosinski_rmq_eq} are equivalent if $\nu$ does not have atoms in some neighbourhood of zero (\cite{asmussen-rosinski}, Proposition 2.1).

\subsection{Estimates for smooth functions}
The errors resulting from Brownian approximation have not been much studied in the literature, at least theoretically. There are some results which we can find in \cite{cont-tankov,cont-voltchkova05}. 
\begin{proposition}\label{oX-2}
Let $X$ be an infinite activity L\'evy process and $t>0$.
\begin{enumerate}
\item If $f\in C^{1}(\R)$ and satisfies $\EE\left|f{'}\left(X_t^{\epsilon}\right)\right|<\infty$, and if there exists $\beta>1$ such that $\left(\sup_{\epsilon\in(0,1]}\EE\left|f{'}\left(X_t^{\epsilon}+\theta \sigma(\epsilon)\hat{W}_t\right)-f{'}\left(X_t^{\epsilon}\right)\right|^{\beta}\right)^{\frac{1}{\beta}}$ and \\ $\left(\sup_{\epsilon\in(0,1]}\EE\left|f^{'}\left(X_t^{\epsilon}+\theta R_t^{\epsilon}\right)-f{'}\left(X_t^{\epsilon}\right)\right|^{\beta}\right)^{\frac{1}{\beta}}$ are finite and integrable with respect to $\theta$ on $[0,1]$, then 
\begin{eqnarray*}
\EE\left(f\left(X_t\right) -f\left(\hat{X}_t^{\epsilon}\right)\right)= o\left(\sigma_0(\epsilon)\right).
\end{eqnarray*}
\item If $f\in C^{2}(\R)$ and satisfies $\EE\left|f{'}\left(X_t^{\epsilon}\right)\right|+\EE\left|f{''}\left(X_t^{\epsilon}\right)\right|<\infty$, and if there exists $\beta>1$ such that $\left(\sup_{\epsilon\in(0,1]}\EE\left|f{''}\left(X_t^{\epsilon}+\theta \sigma(\epsilon)\hat{W}_t\right)-f{''}\left(X_t^{\epsilon}\right)\right|^{\beta}\right)^{\frac{1}{\beta}}$ and\\ $\left(\sup_{\epsilon\in(0,1]}\EE\left|f{''}\left(X_t^{\epsilon}+\theta R_t^{\epsilon}\right)-f{''}\left(X_t^{\epsilon}\right)\right|^{\beta}\right)^{\frac{1}{\beta}}$ are finite and integrable with respect to $\theta$ on $[0,1]$, then 
\begin{eqnarray*}
\EE\left(f\left(X_t\right) -f\left(\hat{X}_t^{\epsilon}\right)\right)=o\left(\sigma_0(\epsilon)^2\right).
\end{eqnarray*}
\end{enumerate}
\end{proposition}

Examples of functions satisfying the above conditions are noted after Proposition~\ref{oX-}.

\begin{proof}
We consider only part 2. The proof for part 1 is similar. By Proposition~\ref{oX-}, we have
\begin{eqnarray*}
\EE\left(f\left(X_t\right) -f\left(X_t^{\epsilon}\right)\right)= \frac{\sigma(\epsilon)^2t}{2}\EE f^{''}\left(X_t^{\epsilon}\right)+o\left(\sigma_0(\epsilon)^2\right).
\end{eqnarray*}
On the other hand, using the same reasoning as in the proof of Proposition~\ref{oX-} (we will replace $R^{\epsilon}$ by $\sigma(\epsilon)\hat{W}$) we get
\begin{eqnarray*}
\EE\left(f(X_t^{\epsilon}+\sigma(\epsilon)\hat{W_t}) -f\left(X_t^{\epsilon}\right)\right)= \frac{\sigma(\epsilon)^2t}{2}\EE f^{''}\left(X_t^{\epsilon}\right)+o\left(\sigma_0(\epsilon)^2\right).
\end{eqnarray*}
Hence
\begin{eqnarray*}
\EE\left(f(X_t)-f(\hat{X}_t^{\epsilon})\right)=o\left(\sigma_0(\epsilon)^2\right).
\end{eqnarray*}
\end{proof}

The combination of Proposition 6.2 of \cite{cont-tankov} and the Spitzer's identity for L\'evy processes (Proposition 1 of \cite{dia-lamberton}) leads to the following result.
\begin{proposition}\label{petitssautsprop}
Let $X$ be an integrable infinite activity L\'evy process. Then
\begin{eqnarray*}
\left|\EE M_t -\EE \hat{M}_t^{\epsilon}\right|\leq 33\sigma(\epsilon)\rho(\epsilon)\left(1+\log \left(\frac{\sqrt{t}}{2\rho(\epsilon)}\right)\right),
\end{eqnarray*}
where $\rho(\epsilon)=\sigma(\epsilon)^{-3}\int_{|x|\leq\epsilon}|x|^3\nu(dx)$.
\end{proposition}

\begin{remark}
Under condition \eqref{asmussen_rosinski_rmq_eq}, we have $\lim_{\epsilon\rightarrow 0}\rho(\epsilon)=0$ and, in turn,
\begin{eqnarray*}
\sigma(\epsilon)\rho(\epsilon)\left(1+\log \left(\frac{\sqrt{t}}{2\rho(\epsilon)}\right)\right)=o\left(\sigma(\epsilon)\right).
\end{eqnarray*}
\end{remark}

\begin{proof}
Let $\delta\in(0,t)$. Using Spitzer's identity for L\'evy processes, we have
\begin{eqnarray*}
\left|\EE M_t -\EE \hat{M}_t^{\epsilon}\right|&=&\left|\int_0^t\frac{\EE X_s^{+}}{s}ds -\int_0^t\frac{\EE (\hat{X}_s^{\epsilon})^{+}}{s}ds\right|
\\&\leq& \int_0^{\delta}\left|\EE X_s^{+}-\EE (\hat{X}_s^{\epsilon})^{+}\right|\frac{ds}{s}+\int_{\delta}^t\left|\EE X_s^{+}-\EE (\hat{X}_s^{\epsilon})^{+}\right|\frac{ds}{s}.
\end{eqnarray*}
On the one hand,
\begin{eqnarray*}
\left|\EE X_s^{+}-\EE (\hat{X}_s^{\epsilon})^{+}\right|&\leq&\EE\left| (X_s^{\epsilon}+R_s^{\epsilon})^{+}- (X_s^{\epsilon}+\sigma(\epsilon)\hat{W}_s)^{+}\right|
\\&\leq&\EE\left| R_s^{\epsilon}-\sigma(\epsilon)\hat{W}_s\right|
\\&\leq& \left(1+\sqrt{\frac{2}{\pi}}\right)\sqrt{s}\sigma(\epsilon).
\end{eqnarray*}
On the other hand, it follows from Proposition 6.2 of \cite{cont-tankov} that
\begin{eqnarray*}
\left|\EE X_s^{+}-\EE (\hat{X}_s^{\epsilon})^{+}\right|\leq A\sigma(\epsilon)\rho(\epsilon),
\end{eqnarray*}
with $A<16.5$ (consider the function $f(x) = x^{+}$). Therefore,
\begin{eqnarray*}
\left|\EE M_t -\EE \hat{M}_t^{\epsilon}\right|&\leq&2\left(1+\sqrt{\frac{2}{\pi}}\right)\sigma(\epsilon)\sqrt{\delta}+A\sigma(\epsilon)\rho(\epsilon)\log \left(\frac{t}{\delta}\right)
\\&\leq&16.5\sigma(\epsilon)\left(\sqrt{\delta}+\rho(\epsilon)\log \left(\frac{t}{\delta}\right)\right).
\end{eqnarray*}
The last expression is minimal for $\delta=4\rho(\epsilon)^2$, and so the desired result follows by substitution.
\end{proof}

\subsection{Estimates by Skorokhod embedding}
We will use a powerful tool to prove the results of this section. This is the Skorokhod embedding theorem. We will begin by defining some useful notations.
\begin{definition}
We define
\begin{eqnarray*}
&&\beta(\epsilon)=\frac{\int_{|x|\leq\epsilon}x^4\nu(dx)}{\left(\sigma_0(\epsilon)\right)^4}, \ \beta_{p,\theta}^t(\epsilon)=\beta(\epsilon)^{\frac{p\theta}{p+4\theta}}\left[\left(\log\left(\frac{t}{\beta(\epsilon)^{\frac{2\theta}{p+4\theta}}}+3\right)\right)^p+1\right],
\\&&\beta_1^t(\epsilon)=\beta(\epsilon)^{\frac{1}{6}}\left(\sqrt{\log\left(\frac{t}{\beta(\epsilon)^{\frac{1}{3}}}+3\right)}+1\right), \ \beta_2^t(\epsilon)=\beta(\epsilon)^{\frac{1}{4}}\left(\log\left(\frac{t}{\beta(\epsilon)^{\frac{1}{4}}}+3\right)+1\right).
\end{eqnarray*}
\end{definition}

\begin{remark}
Note that under condition \eqref{asmussen_rosinski_rmq_eq}, we have $\lim_{\epsilon\rightarrow 0}\beta(\epsilon)=0$.
\end{remark}

The proof of Proposition~\ref{petitssautsprop} cannot be extended to the Lipschitz functions, because the reformulation of the Spitzer identity for L\'evy processes cannot be applied in that case. We have to use another method. Define
\begin{eqnarray*}
 V_{j,n}=R^{\epsilon}_{\frac{jt}{n}}-R^{\epsilon}_{(j-1)\frac{t}{n}},
\end{eqnarray*}
$j=1,\dots,n$, so that $R^{\epsilon}_{kt/n}=\sum_{j=1}^k V_{j,n}$, $k=1,\dots,n$. The $V_{j,n}$ are i.i.d. with the same distribution as $R^{\epsilon}_{t/n}$, hence $\EE\left(V_{j,n}\right)=0$ and $\mbox{Var}\left(V_{j,n}\right)=\sigma(\epsilon)^{2}t/n$. Thus, by Skorokhod's embedding theorem (Theorem 1 of \cite{skorokhod}, see p. 163), there exist positive i.i.d. r.v.'s  $\tau_j$, $j=1,\dots,n$, and a standard Brownian motion, $\hat{B}$, such that the (partial sums)
$R^{\epsilon}_{kt/n}$ and the $\hat{B}_{\tau_1+\dots+\tau_k}$, $k=1,\dots,n$, have the same joint distributions; moreover, $\EE\left( \tau_1\right)=  \mbox{Var}\left(V_{1,n}\right)$ and 
\begin{equation}\label{moments_tau}
\EE \tau_1^2\leq 4\EE V_{1,n}^{4}.
\end{equation}
Further, note that the $\sigma(\epsilon)\hat{W}_{kt/n}$ and $\hat{B}_{\sigma(\epsilon)^2kt/n}$, $k=1,\dots,n$, have
the same joint distributions. Set
\begin{eqnarray*}
T_k=\tau_1+\dots+\tau_k,\ T_k^{\epsilon}=\frac{\sigma(\epsilon)^{2}kt}{n}.
\end{eqnarray*}
This setting will be used in all of the subsequent results.

\begin{theorem}\label{petitssauts}
Let $X$ be an integrable infinite activity L\'evy process, and $f$ a Lipschitz function. Then
\begin{eqnarray*}
\left|\EE f\left(M_t\right) -\EE f\left(\hat{M}_t^{\epsilon}\right)\right|\leq C_t\sigma_0(\epsilon)\beta_1^t(\epsilon),
\end{eqnarray*}
where $C_t$ is a positive constant independent of $\epsilon$.
\end{theorem}

\begin{proof}
Set
\begin{eqnarray*}
I^{\epsilon}_f&=&\left|\EE \left(f\left(\sup_{0 \leq s\leq t} X_{s}\right) -f\left(\sup_{0 \leq s \leq t} \left(X_{s}^{\epsilon}+\sigma(\epsilon)\hat{W}_{s}\right)\right)\right)\right|
\\I^{\epsilon}_f(n)&=&\left|\EE \left(f\left(\sup_{0 \leq k \leq n} X_{\frac{kt}{n}}\right) -f\left(\sup_{0 \leq k \leq n} \left(X_{\frac{kt}{n}}^{\epsilon}+\sigma(\epsilon)\hat{W}_{\frac{kt}{n}}\right)\right)\right)\right|.
\end{eqnarray*}
Because $f$ is, say, $K$-Lipschitz, we can show that
\begin{eqnarray*}
\left|f\left(\sup_{0 \leq k \leq n} X_{\frac{kt}{n}}\right) -f\left(\sup_{0 \leq k \leq n} \left(X_{\frac{kt}{n}}^{\epsilon}+\sigma(\epsilon)\hat{W}_{\frac{kt}{n}}\right)\right)\right|&\leq&K\left(\sup_{0 \leq s \leq t} \left|R^{\epsilon}_{s}\right|+\sigma(\epsilon)\sup_{0 \leq s \leq t}\left|\hat{W}_{s}\right|\right).
\end{eqnarray*}
As the right hand side expression is integrable, by dominated convergence we can deduce that $\lim_{n\rightarrow +\infty}I^{\epsilon}_f(n)=I^{\epsilon}_f$.
It holds that
\begin{eqnarray*}
I^{\epsilon}_f(n)&=&\left|\EE \left(f\left(\sup_{0 \leq k \leq n} \left(X_{\frac{kt}{n}}^{\epsilon}+\hat{B}_{T_k}\right)\right) -f\left(\sup_{0 \leq k \leq n} \left(X_{\frac{kt}{n}}^{\epsilon}+\hat{B}_{T_k^{\epsilon}}\right)\right)\right)\right|
\\&\leq&K\EE\left|\sup_{0 \leq k \leq n} \left(X_{\frac{kt}{n}}^{\epsilon}+\hat{B}_{T_k}\right)-\sup_{0 \leq k \leq n} \left(X_{\frac{kt}{n}}^{\epsilon}+\hat{B}_{T_k^{\epsilon}}\right)\right|
\\&\leq&K\EE\sup_{1 \leq k \leq n} \left|\hat{B}_{T_k}-\hat{B}_{T_k^{\epsilon}}\right|.
\end{eqnarray*}
Part 1 of the following theorem concludes the proof.
\end{proof} 

\begin{theorem}\label{plongement}
Let $X$ be an infinite activity L\'evy process. Then:
\begin{itemize}
\item It holds that
\begin{equation*}
\limsup_{n\rightarrow +\infty}\EE\sup_{1 \leq k \leq n} \left|\hat{B}_{T_k}-\hat{B}_{T_k^{\epsilon}}\right|\leq C_t\sigma_0(\epsilon)\beta_1^t(\epsilon).
\end{equation*}
\item It holds that
\begin{equation*}
\limsup_{n\rightarrow +\infty}\EE\sup_{1 \leq k \leq n} \left|\hat{B}_{T_k}-\hat{B}_{T_k^{\epsilon}}\right|^2\leq C_t\sigma_0(\epsilon)^2\beta_2^t(\epsilon).
\end{equation*}
\item For any reals $p\geq1$ and $\theta\in(0,1)$, it holds that
\begin{equation*}
\limsup_{n\rightarrow +\infty}\EE\sup_{1 \leq k \leq n} \left|\hat{B}_{T_k}-\hat{B}_{T_k^{\epsilon}}\right|^p\leq C_{p,\theta,t}\sigma_0(\epsilon)^p\beta_{p,\theta}^t(\epsilon).
\end{equation*}
\end{itemize}
In the above, $C_t$ and $C_{p,\theta,t}$ are constants independent of $\epsilon$.
\end{theorem}

This theorem is the main result of this section.
\begin{lemma}\label{Pnlemme}
Let X be an infinite activity L\'evy process. Then, for any $\delta>0$,
\begin{eqnarray*}
\limsup_{n\rightarrow +\infty}\P\left[\sup_{1 \leq k \leq n} \left|T_k-T_k^{\epsilon}\right|>\delta\right]&\leq&\frac{4t\sigma_0(\epsilon)^4\beta(\epsilon)}{\delta^2}.
\end{eqnarray*}
\end{lemma}

\begin{proof}
As $T_k-T_k^{\epsilon}=\sum_{i=1}^k \left(\tau_i-\EE\left(\tau_i\right)\right)$, by Kolmogorov's inequality
\begin{eqnarray*}
\P\left[\sup_{1 \leq k \leq n} \left|T_k-T_k^{\epsilon}\right|>\delta\right]&\leq& \frac{\mbox{Var}\left(T_n-T_n^{\epsilon}\right)}{\delta^2}\\
&\leq&\frac{n\mbox{Var}\left(\tau_1\right)}{\delta^2}\\
&\leq&\frac{n\EE\tau_1^2}{\delta^2}\\
&\leq& \frac{4n\EE \left(R_{\frac{t}{n}}^{\epsilon}\right)^4}{\delta^2},
\end{eqnarray*}
where the last inequality follows from \eqref{moments_tau}. The proof then follows from Proposition~\ref{Rnprop}.
\end{proof}

\begin{proof}[Proof of Theorem~\ref{plongement}]
For $\delta>0$, we have
\begin{eqnarray*}
\EE\sup_{1 \leq k \leq n} \left|\hat{B}_{T_k}-\hat{B}_{T_k^{\epsilon}}\right|&=&I_1+I_2,
\end{eqnarray*}
with
\begin{eqnarray*}
&&I_1=\EE\sup_{1 \leq k \leq n} \left|\hat{B}_{T_k}-\hat{B}_{T_k^{\epsilon}}\right|\indicatrice_{\{\sup_{1\leq k\leq n}\left|T_k-T_k^{\epsilon}\right|\leq\delta\}}
\\&&I_2=\EE\sup_{1 \leq k \leq n} \left|\hat{B}_{T_k}-\hat{B}_{T_k^{\epsilon}}\right|\indicatrice_{\{\sup_{1\leq k\leq n}\left|T_k-T_k^{\epsilon}\right|>\delta\}}.
\end{eqnarray*}
On $\{\sup_{1\leq k\leq n}\left|T_k-T_k^{\epsilon}\right|\leq\delta\}$, set, for $k$ fixed,
\begin{eqnarray*}
&&s_1=T_k^{\epsilon}\wedge T_k\\
&&s_2=T_k^{\epsilon}\vee T_k.
\end{eqnarray*}
We have $s_1\leq s_2\leq s_1+\delta$. Let $j$ be such that $j\delta\leq s_1<(j+1)\delta$. We have $s_1\leq s_2\leq (j+2)\delta$. If $j\delta\leq s_1\leq s_2\leq (j+1)\delta$, we have
\begin{eqnarray*}
\left|\hat{B}_{s_1}-\hat{B}_{s_2}\right|&\leq& \left|\hat{B}_{s_1}-\hat{B}_{j\delta}\right|+\left|\hat{B}_{j\delta}-\hat{B}_{s_2}\right|\\
&\leq&2\sup_{0\leq j\leq \left[\frac{\sigma(\epsilon)^2t}{\delta}\right]+1}\left(\sup_{j\delta\leq u\leq(j+1)\delta}\left|\hat{B}_u-\hat{B}_{j\delta}\right|\right).
\end{eqnarray*}
If $j\delta\leq s_1\leq (j+1)\delta\leq s_2\leq (j+2)\delta$, we have
\begin{eqnarray*}
\left|\hat{B}_{s_1}-\hat{B}_{s_2}\right|&\leq& \left|\hat{B}_{s_1}-\hat{B}_{j\delta}\right|+\left|\hat{B}_{j\delta}-\hat{B}_{(j+1)\delta}\right|+\left|\hat{B}_{(j+1)\delta}-\hat{B}_{s_2}\right|\\
&\leq&3\sup_{0\leq j\leq \left[\frac{\sigma(\epsilon)^2t}{\delta}\right]+2}\left(\sup_{j\delta\leq u\leq(j+1)\delta}\left|\hat{B}_u-\hat{B}_{j\delta}\right|\right).
\end{eqnarray*}
Hence
\begin{eqnarray*}
I_1
&\leq&3\EE\sup_{0\leq j\leq \left[\frac{\sigma(\epsilon)^2t}{\delta}\right]+2}\left(\sup_{j\delta\leq u\leq(j+1)\delta}\left|\hat{B}_u-\hat{B}_{j\delta}\right|\right)\\
&=&3\EE\sup_{1\leq j\leq \left[\frac{\sigma(\epsilon)^2t}{\delta}\right]+3}\left(\sup_{(j-1)\delta\leq u\leq j\delta}\left|\hat{B}_u-\hat{B}_{(j-1)\delta}\right|\right).
\end{eqnarray*}
The r.v.'s $\left(\sup_{(j-1)\delta\leq u\leq j\delta}\left|\hat{B}_u-\hat{B}_{(j-1)\delta}\right|\right)_{1\leq j\leq \left[\frac{\sigma(\epsilon)^2t}{\delta}\right]+3}$ are i.i.d. with the same distribution as $\sup_{0\leq u\leq \delta}\left|\hat{B}_u\right|$ and, in turn, $\sqrt{\delta}\sup_{0\leq u\leq 1}\left|\hat{B}_u\right|$. Then
\begin{eqnarray*}
I_1
&\leq&3\sqrt{\delta}\EE\sup_{1\leq j\leq \left[\frac{\sigma(\epsilon)^2t}{\delta}\right]+3}V_j,
\end{eqnarray*}
where $\left(V_j\right)_{1\leq j\leq \left[\frac{\sigma(\epsilon)^2t}{\delta}\right]+3}$ are i.i.d. r.v.'s with the same distribution as\\ $\sup_{0\leq u\leq 1}\left|\hat{B}_u\right|$. On the other hand, we know that if $\left(V_j\right)_{1\leq j\leq m}$ are i.i.d. r.v.'s satisfying $\EE e^{\alpha V_1^2}<\infty$ where $\alpha$ is a positive real, then
\begin{eqnarray*}
\EE\sup_{1\leq j\leq m}V_j\leq g\left(m\EE e^{\alpha V_1^2}\right),
\end{eqnarray*}
where $g : x\in[1,+\infty)\rightarrow \sqrt{\frac{1}{\alpha}\log(x)}$. Indeed, since $g$ is concave, we have
\begin{eqnarray*}
\EE\sup_{1\leq j\leq m}V_j&=&\EE\sup_{1\leq j\leq m}g\left(e^{\alpha V_j^2}\right)\\
&=&\EE g\left(\sup_{1\leq j\leq m}e^{\alpha V_j^2}\right), \ \mbox{because g is non-decreasing}\\
&\leq&g\left(\EE\sup_{1\leq j\leq m}e^{\alpha V_j^2}\right), \ \mbox{by Jensen's inequality}\\
&\leq&g\left(\EE\sum_{j=1}^me^{\alpha V_j^2}\right), \ \mbox{because g is non-decreasing}\\
&=& g\left(m\EE e^{\alpha V_1^2}\right).
\end{eqnarray*}
In our case $V_1=\sup_{0\leq u\leq 1}\left|\hat{B}_u\right|$. So
\begin{eqnarray*}
&&V_1\leq \sup_{0\leq u\leq 1}\hat{B}_u+\sup_{0\leq u\leq 1}\left(-\hat{B}_u\right).
\end{eqnarray*}
For $\alpha\in(0,1/8)$, we have
\begin{eqnarray*}
\EE e^{\alpha V_1^2}&\leq& \EE e^{2\alpha\left(\left(\sup_{0\leq u\leq 1}\hat{B}_u\right)^2+\left(\sup_{0\leq u\leq1}\left(-\hat{B}_u\right)\right)^2\right)} \\
&\leq& \left(\EE e^{4\alpha\left(\sup_{0\leq u\leq 1}\hat{B}_u\right)^2}\right)^{\frac{1}{2}} \left(\EE e^{4\alpha\left(\sup_{0\leq u\leq1}\left(-\hat{B}_u\right)\right)^2}\right)^{\frac{1}{2}}\\
&=& \EE e^{4\alpha\left(\sup_{0\leq u\leq 1}\hat{B}_u\right)^2}
\\&=&(1-8\alpha)^{-\frac{1}{2}}.
\end{eqnarray*}
The last equality follows from $\left(\sup_{0\leq u\leq 1}\hat{B}_u\right)^2\sim\chi_1^2$ upon using the moment-generating
function of the $\chi_1^2$ distribution, given by $(1-2\beta)^{-\frac{1}{2}}$ for $\beta<\frac{1}{2}$.

It follows straightforwardly from the above that, for $\alpha\in(0,\frac{1}{8})$,
\begin{eqnarray*}
I_1&\leq&C_{\alpha}\sqrt{\delta}\sqrt{\log\left(\frac{\sigma(\epsilon)^2t}{\delta}+3\right)},
\end{eqnarray*}
where $C_{\alpha}=3\sqrt{\frac{1}{\alpha}\left(1-\frac{\log(1-8\alpha)}{2\log(3)}\right)}$. Consider now $I_2$. We have \small
\begin{eqnarray*}
I_2&\leq&\left(\EE\left(\sup_{1 \leq k \leq n} \left|\hat{B}_{T_k}-\hat{B}_{T_k^{\epsilon}}\right|\right)^2\right)^{\frac{1}{2}}\left(\P\left[\sup_{1 \leq k \leq n} \left|T_k-T_k^{\epsilon}\right|>\delta\right]\right)^\frac{1}{2}\\ 
&\leq&\left(\EE\left(\sup_{1 \leq k \leq n} \left|\hat{B}_{T_k}\right|+\sup_{1 \leq k \leq n}\left|\hat{B}_{T_k^{\epsilon}}\right|\right)^2\right)^{\frac{1}{2}}\left(\P\left[\sup_{1 \leq k \leq n} \left|T_k-T_k^{\epsilon}\right|>\delta\right]\right)^\frac{1}{2}
\\&\leq&\left(\left(\EE\sup_{0\leq s \leq t} \left|R^{\epsilon}_{s}\right|^2\right)^{\frac{1}{2}}+\left(\EE\sup_{0 \leq s \leq \sigma(\epsilon)^2t} \left|\hat{B}_{s}\right|^2\right)^{\frac{1}{2}}\right)\left(\P\left[\sup_{1 \leq k \leq n} \left|T_k-T_k^{\epsilon}\right|>\delta\right]\right)^\frac{1}{2}
\\&\leq&2\left(\left(\EE\left|R^{\epsilon}_{t}\right|^2\right)^{\frac{1}{2}}+\left(\EE\left|\hat{B}_{\sigma(\epsilon)^2t}\right|^2\right)^{\frac{1}{2}}\right)\left(\P\left[\sup_{1 \leq k \leq n} \left|T_k-T_k^{\epsilon}\right|>\delta\right]\right)^\frac{1}{2}
\\&\leq&4\sqrt{t}\sigma(\epsilon)\left(\P\left[\sup_{1 \leq k \leq n} \left|T_k-T_k^{\epsilon}\right|>\delta\right]\right)^\frac{1}{2},
\end{eqnarray*}
where the fourth inequality is obtained using Doob's inequality.
So, by Lemma~\ref{Pnlemme}, we have
\begin{eqnarray*}
\limsup_{n\rightarrow +\infty}I_2&\leq& 4\sqrt{t}\sigma(\epsilon)\left(\frac{4t\sigma_0(\epsilon)^4\beta(\epsilon)}{\delta^2}\right)^{\frac{1}{2}}.
\end{eqnarray*}
Hence \small
\begin{eqnarray*}
\limsup_{n\rightarrow +\infty}\EE\sup_{1 \leq k \leq n} \left|\hat{B}_{T_k}-\hat{B}_{T_k^{\epsilon}}\right|&\leq&C_{\alpha}\sqrt{\delta\log\left(\frac{\sigma(\epsilon)^2t}{\delta}+3\right)}+\frac{8t}{\delta}\sigma(\epsilon)\sigma_0(\epsilon)^2\sqrt{\beta(\epsilon)}.
\end{eqnarray*}\normalsize
Part 1 now follows by letting $C_{t}=\max\left(C_{\alpha},8t\right)$ and choosing $\delta=\sigma_0(\epsilon)^2\beta(\epsilon)^{\frac{1}{3}}$. 

For the proof of parts 2 and 3 of the theorem, we refer the reader to [\cite{dia}, pp. 86-89]. However, some small corrections are needed in the proof of part 3 in order to comply with the definition of $\beta_{p,\theta}^t(\epsilon)$.
\end{proof}

\begin{remark}
Letting $\theta= 1/2$ and $p = 1,$ $2$ in the definition of $\beta_{p,\theta}^t(\epsilon)$, we see that
part 3 of Theorem 3 partially generalizes parts 1 and 2. It may be relevant to note here that for part 3 the proof used the function $g(x)=\left(\alpha^{-1}\log(x)\right)^p$, whereas for parts 1 and 2 it used the function $g(x)=\left(\alpha^{-1}\log(x)\right)^{p/2}$, $p = 1,$ $2$, respectively.
\end{remark}

The following result follows directly from part 1 of Theorem 3. 

\begin{proposition}\label{ctex}
Let $X$ be an integrable infinite activity L\'evy process, and $f$ a Lipschitz function. Then
\begin{eqnarray*}
 \left|\EE f\left(X_t\right)-\EE f\left(\hat{X}_t^{\epsilon}\right)\right|\leq C_t\beta_1^t(\epsilon)\sigma_0(\epsilon),
\end{eqnarray*}
where $C_t$ is a positive constant.
\end{proposition}

\begin{proof}  
We have $R_t^{\epsilon}=^d \hat{B}_{T_n}$, $\sigma(\epsilon)\hat{W}_t=^d \hat{B}_{T_n^{\epsilon}}$. So, if $f$ is $K$-Lipschitz, we have
\begin{eqnarray*}
\left|\EE f\left(X_t\right)-\EE f\left(\hat{X}_t^{\epsilon}\right)\right|&=&\left|\EE f\left(X_t^{\epsilon}+\hat{B}_{T_n}\right)- \EE f\left(X_t^{\epsilon}+\hat{B}_{T_n^{\epsilon}}\right)\right|
\\&\leq& K\EE\left|\hat{B}_{T_n}-\hat{B}_{T_n^{\epsilon}}\right|.
\end{eqnarray*}
We conclude with Theorem~\ref{plongement}.
\end{proof}

For non-Lipschitz functions, we have the following result.
\begin{proposition}\label{oeM}
Let $X$ be an infinite activity L\'evy process and $p>1$. If $\EE e^{pM_t}<\infty$, then for any $x\in\R$ and for any $\theta\in(0,1)$
\begin{eqnarray*}
\left|\EE \left(e^{M_t}-x\right)^{+}-\EE\left(e^{\hat{M}_t^{\epsilon}}-x\right)^{+}\right|\leq C_{p,\theta,t}\sigma_0(\epsilon)\left(\beta_{\frac{p}{p-1},\theta}^t(\epsilon)\right)^{1-\frac{1}{p}},
\end{eqnarray*}
where $C_{p,\theta,t}$ is a positive constant independent of $\epsilon$.
\end{proposition}

\begin{proof}
Define
\begin{eqnarray*}
 M_t^n=\sup_{0\leq k\leq n}\left(X_{\frac{kt}{n}}^{\epsilon}+R_{\frac{kt}{n}}^{\epsilon}\right), \ \hat{M}_t^{\epsilon,n}=\sup_{0\leq k\leq n}\left(X_{\frac{kt}{n}}^{\epsilon}+\sigma(\epsilon)\hat{W}_{\frac{kt}{n}}\right).
\end{eqnarray*}
We know that $\lim_{n\rightarrow +\infty}M_t^n=M_t$ a.s. and $\lim_{n\rightarrow +\infty}\hat{M}_t^{\epsilon,n}=\hat{M}_t^{\epsilon}$ a.s.
Set
\begin{eqnarray*}
 U_t^n=\sup_{0\leq k\leq n}\left(X_{\frac{kt}{n}}^{\epsilon}+\hat{B}_{T_k}\right), \ \hat{U}_t^{\epsilon,n}=\sup_{0\leq k\leq n}\left(X_{\frac{kt}{n}}^{\epsilon}+\hat{B}_{T_k^{\epsilon}}\right).
\end{eqnarray*}
So $M_t^n=^dU_t^n$ and $\hat{M}_t^{\epsilon,n}=^d\hat{U}_t^{\epsilon,n}$.
By the mean value theorem, we have
\begin{eqnarray*}
e^{U_t^n}-e^{\hat{U}_t^{\epsilon,n}}=\left(U_t^n-\hat{U}_t^{\epsilon,n}\right)e^{\bar{U}_t^{\epsilon,n}},
\end{eqnarray*}
where $\bar{U}_t^{\epsilon,n}$ is between $U_t^n$ and $\hat{U}_t^{\epsilon,n}$. Set
\begin{eqnarray*}
I^{\epsilon}_n=\left|\EE \left(e^{U_t^n}-x\right)^{+}-\EE\left( e^{\hat{U}_t^{\epsilon,n}}-x\right)^{+}\right|.
\end{eqnarray*}
Thus
\begin{eqnarray*}
I^{\epsilon}_n&\leq&\EE \left|e^{U_t^n}- e^{\hat{U}_t^{\epsilon,n}}\right|
\\&\leq&\EE\left|U_t^n-\hat{U}_t^{\epsilon,n}\right|e^{\bar{U}_t^{\epsilon,n}}
\\&\leq&\EE\sup_{0\leq k\leq n}\left|\hat{B}_{T_k}-\hat{B}_{T_k^{\epsilon}}\right|e^{\bar{U}_t^{\epsilon,n}}
\\&\leq&\left(\EE\sup_{0\leq k\leq n}\left|\hat{B}_{T_k}-\hat{B}_{T_k^{\epsilon}}\right|^{\frac{p}{p-1}}\right)^{1-\frac{1}{p}}\left(\EE e^{p\bar{U}_t^{\epsilon,n}}\right)^{\frac{1}{p}}
\\&\leq&\left(\EE\sup_{0\leq k\leq n}\left|\hat{B}_{T_k}-\hat{B}_{T_k^{\epsilon}}\right|^{\frac{p}{p-1}}\right)^{1-\frac{1}{p}}\left(\EE \left(e^{pM_t^{n}}+e^{p\hat{M}_t^{\epsilon,n}}\right)\right)^{\frac{1}{p}}
\\&\leq&\left(\EE\sup_{0\leq k\leq n}\left|\hat{B}_{T_k}-\hat{B}_{T_k^{\epsilon}}\right|^{\frac{p}{p-1}}\right)^{1-\frac{1}{p}}\left(\EE \left(e^{pM_t}+e^{p\hat{M}_t^{\epsilon}}\right)\right)^{\frac{1}{p}}.
\end{eqnarray*}
But
\begin{eqnarray*}
\EE \left(e^{pM_t}+e^{p\hat{M}_t^{\epsilon}}\right)&\leq&\EE \left(e^{pM_t}+e^{p\sigma(\epsilon)\sup_{0\leq s\leq t}\hat{W}_s}e^{pM_t^{\epsilon}}\right)
\\&\leq&\EE e^{pM_t}+2e^{\frac{p^2}{2}\sigma(\epsilon)^2t}\EE e^{pM_t^{\epsilon}}
\\&\leq&2e^{\frac{p^2}{2}\sigma(\epsilon)^2t}\EE\left(e^{pM_t}+e^{pM_t^{\epsilon}}\right).
\end{eqnarray*}
So using dominated convergence, Theorem~\ref{plongement} and Lemma~\ref{OeMlemme}, we get\small
\begin{eqnarray*}
\left|\EE \left(e^{M_t}-x\right)^{+}-\EE\left( e^{\hat{M}_t^{\epsilon}}-x\right)^{+}\right|&=&\lim_{n\rightarrow+\infty}\left|\EE \left(e^{M_t^n}-x\right)^{+}-\EE\left( e^{\hat{M}_t^{\epsilon,n}}-x\right)^{+}\right|
\\&=&\limsup_{n\rightarrow+\infty}\left|\EE \left(e^{U_t^n}-x\right)^{+}-\EE\left( e^{\hat{U}_t^{\epsilon,n}}-x\right)^{+}\right|
\\&\leq&C_{p,\theta,t}\sigma_0(\epsilon)\left(\beta_{\frac{p}{p-1},\theta}^t(\epsilon)\right)^{1-\frac{1}{p}}.
\end{eqnarray*}\normalsize
\end{proof}

\subsection{Estimates for cumulative distribution functions}
The bounds obtained in this section are better than those obtained by truncation, provided that condition \eqref{asmussen_rosinski_rmq_eq} is satisfied.
\begin{proposition}\label{oprobXM}
Let $X$ be an infinite activity L\'evy process. Below, the constants $C_t$ and $C_{x,t,q,\theta}$ are independent of $\epsilon$.
\begin{enumerate}
 \item If $b>0$, then
\begin{eqnarray*}
\sup_{x\in\R}\left|\P\left[X_t\geq x\right]-\P\left[\hat{X}_t^{\epsilon}\geq x\right]\right|\leq C_t\sigma_0(\epsilon)\beta_1^t(\epsilon).
\end{eqnarray*}
 \item If $X_t$ has a locally bounded probability density function and $x\in\R$, then for any pair of reals $\theta\in(0,1)$, $q\in(0,1/2]$,
\begin{eqnarray*}
\left|\P\left[X_t\geq x\right]-\P\left[\hat{X}_t^{\epsilon}\geq x\right]\right|\leq C_{x,t,q,\theta}\sigma_0(\epsilon)^{1-q}\left(\beta_{\frac{1}{q}-1,\theta}^t(\epsilon)\right)^{q}.
\end{eqnarray*}
\item If $M_t$ has a locally bounded probability density function on $(0,+\infty)$ and $x>0$, then for any pair of reals $\theta\in(0,1)$, $q\in(0,1/2]$,
\begin{eqnarray*}
\left|\P\left[M_t\geq x\right]-\P\left[\hat{M}_t^{\epsilon}\geq x\right]\right|\leq C_{x,t,q,\theta}\sigma_0(\epsilon)^{1-q}\left(\beta_{\frac{1}{q}-1,\theta}^t(\epsilon)\right)^{q}.
\end{eqnarray*}
\end{enumerate}
\end{proposition}

\begin{proof}
Recall that $R_t^{\epsilon}=^d \hat{B}_{T_n}$ and $\sigma(\epsilon)\hat{W}_t=^d \hat{B}_{T_n^{\epsilon}}$.
Set
\begin{eqnarray*}
Y_t=X_t^{\epsilon}+\hat{B}_{T_n}, \ \hat{Y}_t^{\epsilon}=X_t^{\epsilon}+\hat{B}_{T_n^{\epsilon}}.
\end{eqnarray*}
Thus
\begin{eqnarray*}
\left|\P\left[X_t\geq x\right]-\P\left[\hat{X}_t^{\epsilon}\geq x\right]\right|&=&\left|\P\left[Y_t\geq x\right]-\P\left[\hat{Y}_t^{\epsilon}\geq x\right]\right|
\\&=&\left|\P\left[Y_t\geq x,\hat{Y}_t^{\epsilon}<x\right]-\P\left[Y_t<x,\hat{Y}_t^{\epsilon}\geq x\right]\right|.
\end{eqnarray*}
It holds that
\begin{eqnarray*}
\P\left[Y_t\geq x,\hat{Y}_t^{\epsilon}<x\right]&=&\P\left[x-\left(Y_t-\hat{Y}_t^{\epsilon}\right)\leq \hat{Y}_t^{\epsilon}<x\right]\\
&=&\P\left[x-\left(\hat{B}_{T_n}-\hat{B}_{T_n^{\epsilon}}\right)\leq b B_t+\left(\hat{Y}_t^{\epsilon}-b B_t\right)<x\right].
\end{eqnarray*}
By construction, $b B_t$ is independent of $\left(\hat{Y}_t^{\epsilon}-b B_t\right)$ and of $\left(\hat{B}_{T_n}-\hat{B}_{T_n^{\epsilon}}\right)$. Further, $\frac{1}{b\sqrt{2\pi t}}$ is an upper bound of the probability density function of $b B_t$. By conditioning on the pair $\left(\hat{B}_{T_n}-\hat{B}_{T_n^{\epsilon}},\hat{Y}_t^{\epsilon}-b B_t\right)$, it can thus be concluded that
\begin{eqnarray*}
\P\left[Y_t\geq x,\hat{Y}_t^{\epsilon}<x\right]&\leq&\frac{1}{b\sqrt{2\pi t}}\EE\left|\hat{B}_{T_n}-\hat{B}_{T_n^{\epsilon}}\right|.
\end{eqnarray*}
Analogously, it also holds that
\begin{eqnarray*}
\P\left[Y_t< x,\hat{Y}_t^{\epsilon}\geq x\right]&\leq&\frac{1}{b\sqrt{2\pi t}}\EE\left|\hat{B}_{T_n}-\hat{B}_{T_n^{\epsilon}}\right|.
\end{eqnarray*}
We get the first part of the proposition by using Theorem~\ref{plongement}. 

We now prove the second part of the proposition. Let $p\geq1$. By Lemma~\ref{probX-Y}, there exists $K_{x,t}>0$ such that, for any $\delta>0$,
\begin{eqnarray*}
\left|\P\left[Y_t\geq x\right]-\P\left[\hat{Y}_t^{\epsilon}\geq x\right]\right|
&\leq& K_{x,t}\delta+\frac{\EE\left|Y_t-\hat{Y}_t^{\epsilon}\right|^p}{\delta^p}\\
&=&K_{x,t}\delta+\frac{\EE\left|\hat{B}_{T_n}-\hat{B}_{T_n^{\epsilon}}\right|^p}{\delta^p}.
\end{eqnarray*}
Hence, given $\theta\in(0,1)$, by Theorem~\ref{plongement} there exists a constant $C_{p,\theta,t}>0$ such that
\begin{eqnarray*}
\left|\P\left[Y_t\geq x\right]-\P\left[\hat{Y}_t^{\epsilon}\geq x\right]\right|
&\leq&K_{x,t}\delta+C_{p,\theta,t}\frac{\sigma_0(\epsilon)^p\beta_{p,\theta}^t(\epsilon)}{\delta^p}.
\end{eqnarray*}
Choosing $\delta=\sigma_0(\epsilon)^{\frac{p}{p+1}}\beta_{p,\theta}^t(\epsilon)^{\frac{1}{p+1}}$ yields 
\begin{eqnarray*}
\left|\P\left[Y_t\geq x\right]-\P\left[\hat{Y}_t^{\epsilon}\geq x\right]\right|&\leq&2\max\left(K_{x,t},C_{p,\theta,t}\right)\sigma_0(\epsilon)^{\frac{p}{p+1}}\beta_{p,\theta}^t(\epsilon)^{\frac{1}{p+1}}.
\end{eqnarray*}
The result then follows by substituting $p = 1/q - 1$.

For the third part of the proposition, we use the notation of Proposition~\ref{oeM}. Note that
\begin{eqnarray*}
\left|\P\left[M_t\geq x\right]-\P\left[\hat{M}_t^{\epsilon}\geq x\right]\right|&=&\lim_{n\rightarrow\infty}\left|\P\left[M_t^n\geq x\right]-\P\left[\hat{M}_t^{\epsilon,n}\geq x\right]\right|
\\&=&\lim_{n\rightarrow\infty}\left|\P\left[U_t^n\geq x\right]-\P\left[\hat{U}_t^{\epsilon,n}\geq x\right]\right|.
\end{eqnarray*}
Let $p\geq1$ and put $I_{x,\delta} = [x-\delta, x+\delta)$. Using the proof of Lemma~\ref{probX-Y}, we have
\begin{eqnarray*}
\left|\P\left[U_t^n\geq x\right]-\P\left[\hat{U}_t^{\epsilon,n}\geq x\right]\right|&\leq&\P\left[U_t^n\in I_{x,\delta}\right] + \frac{\EE\left|U_t^n-\hat{U}_t^{\epsilon,n}\right|^p}{\delta^p}
\\&\leq&\P\left[M_t^n\in I_{x,\delta}\right] + \frac{\EE\sup_{1\leq k\leq n}\left|\hat{B}_{T_k}-\hat{B}_{T_k^{\epsilon}}\right|^p}{\delta^p},
\end{eqnarray*}
for any $\delta>0$. By the assumption on $M_t$, there exists a constant $K^{'}_{x,t}>0$ such that $\P\left[M_t\in I_{x,\delta}\right]<K^{'}_{x,t}\delta$ for any $\delta>0$. Combined with Theorem~\ref{plongement}, letting $n\rightarrow\infty$ yields
\begin{eqnarray*}
\lim_{n\rightarrow\infty}\left|\P\left[U_t^n\geq x\right]-\P\left[\hat{U}_t^{\epsilon,n}\geq x\right]\right|&\leq&K^{'}_{x,t}\delta+ C_{p,\theta,t}\frac{\sigma_0(\epsilon)^p\beta_{p,\theta}^t(\epsilon)}{\delta^p},
\end{eqnarray*}
for some constant $C_{p,\theta,t}>0$. So as in part 2, the result follows by choosing $\delta=\sigma_0(\epsilon)^{\frac{p}{p+1}}\beta_{p,\theta}^t(\epsilon)^{\frac{1}{p+1}}$.
\end{proof}



\begin{thebibliography}{10}

\bibitem{asmussen-rosinski} \textsc{Asmussen, S. And Rosinski, J.} (2001). Approximations of small jumps of L\'evy processes with a view towards simulation. {\em J. Appl. Probab.}, \textbf{38}, 482-493.

\bibitem{nielsen} \textsc{Barndorff-Nielsen, O. E.} (1997). Normal inverse Gaussian distributions and stochastic volatility modelling. {\em Scand. J. Statistics}. \textbf{24}, 1-13.

\bibitem{bertoin} \textsc{Bertoin, J.} (1996). {\em L\'evy Processes}. Cambridge University Press.

\bibitem{broadie-yamamoto} \textsc{Broadie, M. And Yamamoto, Y.} (2005). A double-exponential fast Gauss transform algorithm for pricing discrete path-dependent options. {\em Operations Research}. \textbf{53}, 764-779.

\bibitem{cgmy} \textsc{Carr, P. P., Geman, H. And Madan D. B. And Yor, M.} (2002). The fine structure of asset returns: An empirical investigation. {\em Journal of Business}. \textbf{75}, 305-332.

\bibitem{chaumont} \textsc{Chaumont, L.} (2011). On the law of the supremum of L\'evy processes. To appear in {\em Ann. Probab.}

\bibitem{cont-tankov} \textsc{Cont, R. And Tankov, P.} (2004). {\em Financial Modelling with Jump Processes}. Chapman \& Hall/CRC Financial Mathematics Series, Boca Raton.

\bibitem{cont-voltchkova05} \textsc{Cont, R. And Voltchkova, E.} (2005). Integro-differential equations for option prices in exponential L\'evy models. {\em Finance Stoch.} \textbf{9}, 299-325.

\bibitem{dia} \textsc{Dia, E. H. A.} (2010). Options exotiques dans les mod\`eles exponentiels de L\'evy. Doctoral thesis, Universit\'e Paris-Est. Available at http://tel.archives-ouvertes.fr/tel-00520583/en/.

\bibitem{dia-lamberton} \textsc{Dia, E. H. A. And Lamberton, D.} (2011). Connecting discrete and continuous lookback or hindsight options in exponential L\'evy models. {\em Advances in Applied Probability}. \textbf{43}, 1136-1165.

\bibitem{eberlein01} \textsc{Eberlein, E.} (2001). Application of generalized hyperbolic L\'evy motions to finance. In {\em L\'evy Processes: Theory and Applications}, eds OE Barndorff-Nielsen, T Mikosch \& S Resnick, Birkh\"auser Verlag, pp. 319-337.

\bibitem{feng-linetsky08} \textsc{Feng, L. And Linetsky, V.} (2008). Pricing discretely monitored barrier options and defaultable bonds in L\'evy process models: a fast Hilbert transform approach. {\em Math. Finance}. \textbf{18}, 337-384.

\bibitem{feng-linetsky} \textsc{Feng, L. And Linetsky, V.} (2009). Computing exponential moments of the discrete maximum of a L\'evy process and lookback Options. {\em Finance Stoch}. \textbf{13}, 501-529.

\bibitem{morris} \textsc{Morris, C. N.} (1982). Natural exponential families with quadratic variance functions. {\em Ann. Statist}. \textbf{10}, 65-80.

\bibitem{petrella-kou} \textsc{Petrella, G. And Kou, S. G.} (2004). Numerical pricing of discrete barrier and lookback options via Laplace transforms. {\em Journal of Computational Finance}. \textbf{8}, 1-37.

\bibitem{rydberg} \textsc{Rydberg, T. H.} (1997). The normal inverse gaussian lévy process: simulation and approximation. {\em Stoch. Models}. \textbf{13}, 887-910.

\bibitem{sato} \textsc{Sato, K.} (1999). {\em L\'evy processes and Infinitely Divisible Distributions}. Cambridge university press.

\bibitem{signahl} \textsc{Signahl, M.} (2003). On error rates in Normal approximations and simulation schemes for Lévy processes. {\em Stoch. Models}.  \textbf{19}, 287-298.

\bibitem{skorokhod} \textsc{Skorokhod, A. V.} (1965). {\em Studies in the Theory of Random Processes}. Addison-Wesley, Reading, Mass. 

\bibitem{tankov} \textsc{Tankov, P.} (2004). L\'evy processes in finance: inverse problems and dependence modelling. Doctoral Thesis, Ecole Polytechnique.
\end{thebibliography}
\end{document}